\documentclass{ecai} 

\usepackage{latexsym}
\usepackage{amssymb}
\usepackage{amsmath}
\usepackage{amsthm}
\usepackage{booktabs}
\usepackage{enumitem}
\usepackage{graphicx}
\usepackage{color}
\usepackage{algorithm}
\usepackage{algorithmic}
\usepackage{booktabs}
\usepackage{multirow}
\usepackage{svg}
\usepackage{comment}
\usepackage{amsmath,amssymb,mathtools}
\usepackage{siunitx}
\usepackage{enumitem}
\setlist[itemize]{noitemsep, topsep=0pt}
\usepackage{amsthm}

\newtheorem{theorem}{Theorem}
\newtheorem{lemma}[theorem]{Lemma}

\newtheorem{proposition}[theorem]{Proposition}

\newtheorem{definition}{Definition}

\newcommand{\BibTeX}{B\kern-.05em{\sc i\kern-.025em b}\kern-.08em\TeX}

\newcommand{\Er}{\mathbf{E}}

\newcommand{\Ac}{\ensuremath{\mathcal{A}}}

\newcommand{\Gc}{\ensuremath{\mathcal{G}}}

\newcommand{\Nc}{\ensuremath{\mathcal{N}}}

\newcommand{\Rc}{\ensuremath{\mathcal{R}}}

\newcommand{\Xc}{\ensuremath{\mathcal{X}}}

\newcommand{\lsb}{\ensuremath{\left[}}
\newcommand{\rsb}{\ensuremath{\right]}}

\newcommand{\fref}[1]{\fref{#1}}

\newcommand{\cref}[1]{Chapter~\ref{#1}}

\newtheorem{problem}{Problem}

\DeclareMathOperator*{\argmax}{arg\,max}

\begin{document}


\begin{frontmatter}


\paperid{1338} 


\title{United We Stand: Decentralized Multi-Agent \\Planning With Attrition}


\author[A]{\fnms{Nhat}~\snm{Nguyen}\thanks{Corresponding Author. Email: nhatdaoanh.nguyen@adelaide.edu.au}}
\author[A]{\fnms{Duong}~\snm{Nguyen}}
\author[B]{\fnms{Gianluca}~\snm{Rizzo}}
\author[A]{\fnms{Hung}~\snm{Nguyen}} 

\address[A]{The University of Adelaide, Australia}
\address[B]{HES SO Valais, Switzerland, and the University of Foggia, Italy.}


\begin{abstract}
Decentralized planning is a key element of cooperative multi-agent systems for information gathering tasks. However, despite the high frequency of agent failures in realistic large deployment scenarios, current approaches perform poorly in the presence of failures, by not converging at all, and/or by making very inefficient use of resources (e.g. energy). In this work, we propose Attritable MCTS (A-MCTS), a decentralized MCTS algorithm capable of timely and efficient adaptation to changes in the set of active agents. It is based on the use of a global reward function for the estimation of each agent's local contribution, and regret matching for coordination. We evaluate its effectiveness in realistic data-harvesting problems under different scenarios. We show both theoretically and experimentally that A-MCTS enables efficient adaptation even under high failure rates. Results suggest that, in the presence of frequent failures, our solution improves substantially over the best existing approaches in terms of global utility and scalability.
\end{abstract}

\end{frontmatter}


\section{Introduction}

Cooperative multi-agent systems (MAS) are systems where multiple agents (such as autonomous vehicles/drones) work together to achieve a common goal such as maximizing a shared utility \citep{dorri2018multi}. These agents can communicate and coordinate with each other, either directly or indirectly, to solve complex tasks that are beyond a single agent's capabilities. Examples are drone swarms for autonomous aerial surveillance or disaster relief operations, or teams of robots that collaborate to explore unknown environments, harvest data from sensors, or manipulate objects, among others \citep{xie2017multi}.s.

Centralized approaches for addressing the MAS planning problem do not scale with the number of agents, as the amount of computational resources required to solve it quickly becomes prohibitive. In addition, the amount of information exchange required for a centralized planner to manage all agents can be unfeasibly high in large-scale settings, particularly in remote areas and disaster scenarios \citep{trodden2009robust}. Thus, recent research has focused on decentralized approaches for online MAS planning~\citep{zhang2021multi}. Indeed, they offer enhanced robustness, reduced computational burden, and lower communication load, particularly on infrastructure-based communications such as cellular radio access networks ~\citep{claes2017decentralised}.. 

The main challenge in decentralized approaches for cooperative MAS is optimizing agents' actions in a distributed manner to maximize a global reward function. This problem can typically be modeled as a Decentralized Partially Observable Markov Decision Process (Dec-POMDP)~\citep{oliehoek2016concise}. However, the computational complexity of Dec-POMDPs presents a significant hurdle, making direct optimal solution search infeasible in polynomial-time \citep{bernstein2002complexity}.
Several sampling-based planning algorithms have also been proposed to improve computational efficiency, particularly for special classes of Dec-POMDPs, such as multi-robot active perception \citep{zhang2021multi}.
A first family of approaches to address this is given by point-based methods, which scale well, but they may not cover the entire belief space well \citep{POINT_pineau2003point,POINT_shani2013survey}. Another set of algorithms is based on policy search \citep{POLICY_seuken2011memory, POLICY_amato2010optimizing} based on a parameterized policy representation. They can handle problems with large action spaces, but they may get stuck in local optima or require many samples
\citep{POLICY_amato2010optimizing}.
Thus, attention has turned towards algorithms based on Monte Carlo tree search (MCTS), due to their ability to effectively explore long planning horizons, their anytime nature~\citep{kocsis2006improved}, and their excellent performance in decentralized settings \citep{claes2017decentralised, best2019dec, czechowski2020decentralized}, effectively overcoming the limitations of other approaches.

In many present-day MAS application scenarios, the departure of agents from the system (henceforth denoted as \textit{attrition}, and due to e.g. failures or energy depletion) is a very common feature. However, all of the main approaches to decentralized MAS planning assume agents are always available and actively contributing to the joint planning process. When applied to scenarios with attrition, they perform in a heavily suboptimal manner and they often do not converge at all \citep{cybenko2021attritable,nguyen2022multi}. In swarm robotics, for instance, agent attrition due to robot failures, damage, or energy depletion affects the overall swarm behavior and task completion. Designing robust algorithms for decentralized MAS planning capable of effectively handling agent loss is critical for their successful deployment many in practical scenarios.

The common approaches for scenarios with attrition are based on periodically resetting agents' learned behavior, and restarting the learning process~\citep {avner2014concurrent,rosenski2016multi,hanawal2021multiplayer}. In volatile settings with frequent failures, such a feature may significantly hamper the overall performance of the active perception task, by slowing down the convergence rate and by keeping the system far from adapting and thus from achieving optimal operating conditions. 
Therefore, how to efficiently and effectively perform online MAS planning in the presence of attrition, while achieving fast convergence, is a key open issue.

In this paper, we develop a novel decentralized planning algorithm that achieves both of these objectives. Our approach is based on MCTS and the use of the global reward instead of the local one in the estimation of each agent's local contribution. Moreover, it exploits regret matching (RM)~\citep{hart2000simple} to coordinate the actions between agents. We prove that our approach guarantees that the average joint action of all agents converges to a Nash equilibrium (NE) if every agent applies the same RM procedure in any cooperative game with a submodular utility function. Arriving at an NE guarantees no diverging interest between the agents, and therefore it ensures that all participants come into a self-enforcing agreement to effectively coordinate their action decisions in a decentralized manner.
The main contributions of our work are:

\begin{itemize}[leftmargin=*]
\item We develop Attritable MCTS (A-MCTS), a new online decentralized planning algorithm, based on MCTS and Regret Matching, that can quickly and efficiently adapt the plan to settings where agents fail, even at high rates.
\item We show that, by modulating the utility function for each agent, under the assumption of submodularity, successive iterations are guaranteed to improve joint policies, and eventually lead to convergence of our algorithm.
We prove a strong convergence result for approximating a pure-strategy Nash equilibrium in a fully distributed fashion.
\item We evaluate our proposed approach in several information-gathering scenarios with attrition. Results suggest that, in the presence of frequent failures, our solution improves substantially over the best existing approaches in terms of global utility and scalability.
\end{itemize}
\section{Related Work}

Information-gathering problems are often modeled as sequential decision-making problems~\citep{yao2020path}. When there are multiple agents, decentralized information gathering can be viewed as a decentralized POMDP~\citep{bernstein2002complexity}. The dominant approach to Dec-POMDP is to first solve the centralized, offline planning over the joint multi-agent policy space, and then push these policies to agents to execute them in a decentralized fashion~\citep{oliehoek2016concise}. When the state of the environment or agents is not known ahead of time, these approaches become infeasible. Fully decentralized Dec-POMDP solvers exist~\citep{spaan2006decentralized}. However, they require significant memory and incur high computational complexity due to the requirements to compute and store all the reachable joint state estimations~\citep{lauri2020multi}.

Recently, simultaneous distributed approaches based on MCTS have gained significant interest due to their flexibility in trading off computation time for accuracy. The key idea is to use the upper confidence bound (UCB) \citep{kocsis2006improved} for planning the best course of action. To implement cooperation between agents, these methods usually keep a predefined model of the teammates, which can be heuristic or machine learning trained \citep{claes2017decentralised,czechowski2020decentralized,choudhury2021scalable}. However, as they are based on trained knowledge, they are unsuitable for online planning in settings that change unpredictably, such as in disastrous environments.
 
To address this, new approaches based not on apriori knowledge about agents' behavior, but on information sharing among them, have been proposed (Dec-MCTS~\citep{best2019dec}). A key aspect of Dec-MCTS and all its subsequent variations \citep{li2019integrating,li2021dec,nguyen2022multi} is that each agent is assigned a local utility function, which does not measure the total team reward but the contribution of that agent only. To deal with any uncertainty that arises during the mission, Dec-MCTS algorithms allow for online replanning during execution, by having agents update their beliefs about the system. 

Under high uncertainty scenarios, a growing body of literature review in the area of multi-drone systems~\citep{williams2014multi, lizzio2022review, dinelli2023configurations, goeckner2023attrition} explores the significant challenges posed by agent attrition -- the loss or removal of individual drones (due to mechanical failures, environmental factors, and human errors), and highlights the need to address attrition for robust system performance. In systems with attrition in which agents may fail abruptly, all of the above-mentioned approaches do not apply, as they do not allow adjusting to attrition in agents' populations. In the present work, we show that the suboptimality of current Dec-MCTS algorithms is due to the usage of the marginal contribution combined with the submodular properties of the global utility functions.

Another body of literature on related works concerns Open Agent Systems (OASYS), in which agents can enter and leave over time. Most solutions to OASYS are either fully or partially offline, i.e., offline planning with online execution \citep{cohen2017open, chandrasekaran2016individual} or online planning with precomputed offline policies \citep{eck2020scalable}. This is not feasible in applications with significant sources of uncertainty, particularly when the environment's state or mental models of the agents are unknown in advance, and when the agents' failures occur abruptly during execution. A recent work \citep{kakarlapudi2022decision} proposed a fully online approach that leverages communication between agents related to their presence to predict the actions of others. This approach assumes that agents can communicate their existence implicitly. In scenarios where the communication is intermittent or agents vanish without warning, the unforeseeable failure can significantly impact coordination and planning, jeopardizing the overall performance. In addition, the computational complexity of modeling each other's presence and predicting their actions can make the system computationally intractable on a large scale. Thus, it isn't directly equipped to handle sudden agent failures and may require further refinement to be viable in large-scale or highly dynamic environments.

Our paper focuses on problems where agents experience hard failures in an abrupt and unforeseeable fashion. This unpredictable nature makes the existing techniques not directly applicable. Therefore, more research is needed to enhance the robustness and resilience of multi-agent systems in such uncertain attrition scenarios. Our work explicitly tackles this challenge by providing a new approach for online decentralized planning for multi-agents that does not require precomputed offline policies and mental models. Instead, agents reason about the actions and existence of others using directly communicated information.
We employ a computational-effective game-based technique to coordinate agents, enabling adaptive decision-making in the presence of peer failures while ensuring fast convergence in polynomial time relative to system size.
\section{Problem Formulation}
\label{sec:problem_formulation}

In this paper, we consider a set $\mathcal{N}$ of $N$ autonomous agents moving within a given area of space. We consider a set of $R$ \textit{regions of interest} in a given area, where $R_k$ is the $k$-th element of the set. We assume each region is a sphere with an equal radius, however, the formulation could extend to more complex models.
Without loss of generality, we assume agents move along an undirected graph $G = (V, E)$ that is placed in the same space as the regions of interest. Each vertex $v_i \in V$ represents a location, and each edge $e_{ij}$ represents a feasible route from vertex $v_i$ to $v_j$. A key property of this graph is that it traverses at least once every region of interest.
This graph typically models constraints to agent trajectories due to the morphology of the monitored environment, presence of obstacles, regions of interest distribution, and characteristics of agent movements, among others. The specific way in which the graph is derived is thus application and context-dependent  \citep{nguyen2022multi}. The graph is defined at the beginning of the mission, it does not change over time and it is known by all agents.

The \emph{path} of agent $n$, denoted as $p^n$, is an ordered list of edges $p^n = (e^n_1, e^n_2,...)$, such that two adjacent edges in the path are connected by a vertex of the graph. With $p = (p^1, ..., p^n, ..., p^N)$ we denote the joint paths of every agent. Let $B$ denote the maximum path length of each agent, which equals the number of edges an agent can traverse. Such a maximum value is derived from the agent's speed, but it may also capture various constraints, e.g. due to finite storage capacity, among others. To any path $p^n$ we associate a cost $b(p^n)$, equal to the number of traversed edges. A region $R_k$ is \textit{observed} if it is traversed by a path of an agent. Every region $R_k$ is associated with a utility $U(R_k)$, which models the value of the information that agents may collect from it. For ease of analytical treatment, we assume that it takes one unit of time to traverse any edge and that any exchange of information among agents is instantaneous. Note however that our approach can be easily extended to account for nonzero exchange duration, as well as for edge traversal times that differ among edges and agents.

Finally, we assume each agent can exchange information at any point in time with any other agent. This models scenarios in which agents have a wireless interface to a cellular access network. We assume the information exchange to be instantaneous, independent from the amount of information shared, and reliable, with no loss. In the experimental section, we relax this assumption and investigate the impact of nonidealities in information sharing on the effectiveness of our approach.

\subsection{Multi-agent planning with attrition}

\label{sec:attrition_settings}

We denote the information-gathering task as a \textit{mission}, for which each agent performs independent actions to achieve a collective goal - maximizing the global utility for the whole team. Each agent $n$ plans its path $p^n$ and coordinates with others in a decentralized manner while satisfying the given budget constraint $B$ on path length.
This formulation of the information gathering problem generalizes many multi-agent path planning problems, such as team orienteering problem~\citep{best2018online}.
We consider a scenario, in which mission planning is performed in a decentralized manner for scalability and computational feasibility, as mentioned in Section 1.
Thus, each agent plans its path while considering the potential actions of other agents and the team's total utility.
We consider scenarios where a subset $\mathcal{F}$ of the N agents fail during the mission. We focus on \textit{hard} failures, where agents interrupt reward collection and information exchange.
We assume the set of agents that fail $\mathcal{F}$ is unknown in advance and the time at which they fail to be determined by any arbitrary criteria or distribution. Therefore, our solution does not rely on knowing its size and probability distribution.
In the occurrence of a failure, all the utility collected by the failed agent is lost, i.e. it is not considered anymore in the computation of the global utility of the mission. This models a typical setup in information gathering, in which data collected by agents is relayed to data sinks only at the end of the mission.

Our goal is to provide an efficient planning and coordination mechanism that can quickly adapt to agent failures and maximize the global utility of the mission within the agent's budget constraint. Let $\mathcal{P} = (P^1, ..., P^n, ..., P^N)$, with $P^n$ denote the set of all possible paths of length $B$ which starts at agent $n$ starting position.
We define the following problem:

\begin{problem}\label{prob:1}(\textit{Multi-agent planning in attrition settings)}
\begin{equation}\label{eq:objectivefunction}
\underset{p \in \mathcal{P}}{\text{maximize}}\   U_g(p)
\end{equation}
	\begin{align}
           \text{Subject to:} &  \quad b(p^n)\leq B,  \quad \forall n \in \mathcal{N}   \label{constraint_pathlength} \\
           & \quad 0 \leq | \mathcal{F} | \leq N
	\end{align}
\end{problem}

Constraint (\ref{constraint_pathlength}) derives from imposing that the total path length for the agent $n$ to be less than the travel budget $B$ available to each agent. Intuitively, our goal is to find a path for each agent such that the global utility associated with all regions observed by all agents during the mission is maximized, while a subset $\mathcal{F}$ of agents fail. Such an optimization problem cannot be solved efficiently. Indeed it is easy to see that Problem 1 is a variant of the well-known NP-hard traveling salesman problem.


\section{Attritable MCTS with Regret Minimization}

In this section, we first give a brief introduction to Monte Carlo Tree Search and its most popular decentralized version. We then show the root cause of the inefficiency of existing decentralized MCTS approaches with attrition. 

MCTS is an excellent approach for online planning problems \citep{kocsis2006improved}. The tree $\mathcal{T}_n$ for agent $n$ is defined such that each node $s$ of the tree represents a state and each edge $a$ starting from that node represents an available action. A branch from the root node to another node represents a valid action sequence. The tree is incrementally grown via a four-step process: \emph{selection}, \emph{expansion}, \emph{rollout}, and \emph{backpropagation}. Decentralized Monte Carlo Tree Search (Dec-MCTS)~\citep{best2019dec} extends the power of MCTS to MAS using intention sharing. Specifically, agent $n$ maintains a probability mass function $q_n(x_n)$ over the set of all possible action sequences $\Xc_n$, where $x_n \in \Xc_n$ is a primitive action sequence. The intentions of other agents except agent $n$ are denoted by $q_{-n}$ and $\Xc_{-n}$. By taking a probabilistic sampling from the communicated intention, each agent can reconstruct the global utility. To create better coordination, rather than optimizing directly for the global utility $U_g$ of the entire team, each agent $n$ instead optimizes for a local \emph{marginal contribution} utility function $U_n$. That is, agent $n$ estimates the rollout score for executing $x_n$ as: 
\begin{equation} 
\label{eq:marginal}
     F_n(x_n) = U_n(x_n, x_{-n}) = U_g(x_n, x_{-n}) - U_g(x_{-n})\ ,
\end{equation}
where $U_g(x_{-n})$ is the global utility without the contribution of agent $n$.

We now analyze Dec-MCTS asymptotic behavior when agents fail. We are particularly interested in \textit{submodular} reward functions, frequently arising in data collection problems~\citep{corah2017efficient,satsangi2018exploiting}. Submodular set function, which is defined in Definition~\ref{def:submodular}, satisfies the diminishing returns property. Regarding the information-gathering problem discussed in this paper (see Section~\ref{sec:problem_formulation}), the marginal gain of adding a new location to the set of visited locations decreases as the number of locations visited increases.

\begin{definition}[Submodular set function]
    \label{def:submodular}
    Let $g: 2^\Omega \rightarrow \mathbb{R}$ be a set function where $2^\Omega$ is the power set of $\Omega$. Then \emph{g} is a submodular function if for every $X, Y \subseteq \Omega$ with $X \subseteq Y$ and every $x \in \Omega \setminus Y$ the following inequality holds
    $$g(X \cup x) - g(X) \ge  g(Y \cup x) - g(Y)\ .$$
\end{definition}

In particular, at iteration $t$, let $x_n$ denote the chosen action sequence of agent $n$ and $x_{-n}$ denote the combined sampled action sequences of other agents. Assume that at the next iteration $t+1$, a subset of agents fails. Let $x^\prime_{-n}$ be the combined sampled action sequences of all agents except agent $n$ and the lost agents (i.e., that is $x^\prime_{-n} \subseteq x_{-n}$).

\begin{proposition}
\label{prop:nondecreasingreward}
If the global objective function $U_g$ is submodular, then $
    F_n^{(t+1)}(x^*_n) \ge F_n^{(t)}(x^*_n)$ by the diminishing return property due to submodularity, where $F_n(x_n)$ is defined in (\ref{eq:marginal}).
\end{proposition}

\begin{figure*}[!ht]
	\begin{center}
 \includegraphics[width=0.7\linewidth]{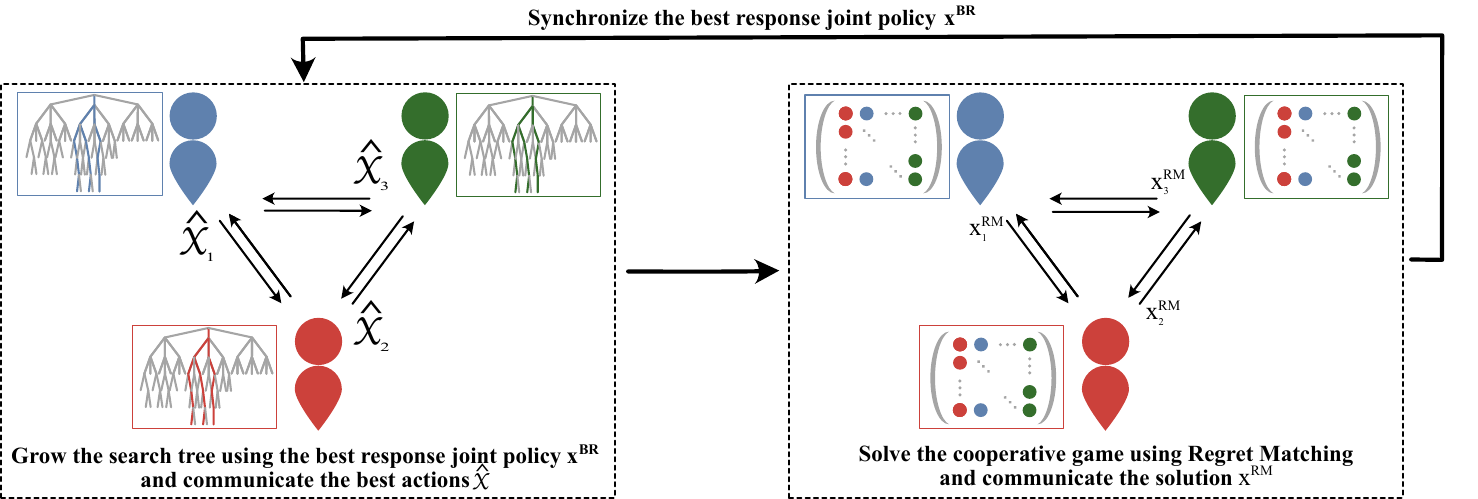}
		\caption{\footnotesize Overview of the A-MCTS algorithm. Agents incrementally grow the search using the best response policy $x^{BR}$ and communicate their best actions $\hat{\mathcal{X}}$. Regret Matching is then used to compute distributively a joint policy for the cooperative game. These solutions are synchronized and the most payoff-dominant is chosen as the best response policy $x^{BR}$.
        \label{fig:algo}}
	\end{center}
\end{figure*}

Proposition~\ref{prop:nondecreasingreward} states that if some agents fail during the mission, the remaining agents would mistakenly perceive that the contribution of their previous actions increases. Hence, they would not be aware of the actual reduction of the global utility and update their plans\footnote{For analysis of Dec-MCTS behavior when agent failures occur after the algorithm has converged, please see Appendix A in the Supplementary Material}. Fixing this issue requires both a new, context-aware way to compute the reward and a new way to coordinate the actions with other agents on this new reward function. In the next section, we present our proposed algorithm for the multi-agent planning in attrition settings problem \ref{prob:1}.

\subsection{Overview of the A-MCTS algorithm}
\label{sec:attritable}
We develop Attritable MCTS (A-MCTS), an online decentralized MCTS algorithm that quickly adapts to agents' attrition and efficiently coordinates action between the remaining agents. Its performance mainly relies on two key factors, including the joint-utility-guided decentralized tree search, and the best response policy given the shared intentions of others.
Each agent runs A-MCTS distributedly to plan for itself a path that is expected to maximize the total utility of the whole mission. Agents then execute the first planned action and observe any changes. After that, they perform replanning from their new state and update the planned paths based on newly available information. The search tree may be pruned by removing all children of the root except the selected branch. This cycle of planning and execution continues until the travel budget expires.
The pseudo-code of A-MCTS for agent $n$ is shown in Algorithm~\ref{alg:A-MCTS}.

The tree $\mathcal{T}_n$ of agent $n$ is incrementally built over its action sequences space $\Xc_n$ while considering the possible behaviors of others $\Xc_{-n}$ (Line 6-9). In the \emph{selection phase}, the discounted upper confidence bound on Tree (D-UCT)~\citep{best2019dec} is applied to handle the abrupt changes in reward values caused by the actions of other agents.

\begin{algorithm}[ht!]
 \caption{A-MCTS algorithm for agent $n$}
 \label{alg:A-MCTS}
 \begin{algorithmic}[1]
 \renewcommand{\algorithmicrequire}{\textbf{Input:}}
 \REQUIRE Global objective function $U_g$, actions budget $B$\\
 \renewcommand{\algorithmicrequire}{\textbf{Output:}}
 \REQUIRE best action sequence $x^\ast_n$ for agent $n$\\
 \STATE $\mathcal{T}_n\ {\leftarrow}  \textit{Initilize MCTS Tree}$
 \STATE \textbf{while} computation budget not met \textbf{do}
 \STATE \quad $\hat{\mathcal{X}_n}\ {\leftarrow} \textit{Select Subset From} (\mathcal{T}_n)$
 \STATE \quad $(\hat{\mathcal{X}}, \hat{\mathcal{X}}_{-n})\ {\leftarrow} \textit{Communicate and Update} (\hat{\mathcal{X}_n})$
 \STATE \quad $x^{BR}\ {\leftarrow} \textit{Regret Matching Coordination} (\hat{\mathcal{X}})$
 \STATE \quad \textbf{for} fixed number of iterations \textbf{do}
 \STATE \quad \quad $x_n\ {\leftarrow} \textit{D-UCT Select, Expand \& Rollout} (\mathcal{T}_n, B)$
 \STATE \quad \quad $F_n\ {\leftarrow} U_g(x_n, x^{BR}_{-n})$ 
 \STATE \quad \quad  $\mathcal{T}_n\ {\leftarrow} Backpropagation (\mathcal{T}_n, F_n)$
 \STATE $x^\ast_n\ {\leftarrow} \textit{Best Next Action} (\mathcal{T}_n)$
 \STATE \textbf{return} $x^\ast_n$
 \end{algorithmic}
\end{algorithm}

The key idea of our proposed algorithm is to have the search trees of every agent be guided by the same utility of the joint action sequences. This is achieved by letting all agents optimize their local actions using the global utility $U_g$ directly (Line 8). Each agent can then decide its path $x_n$ independently to maximize $U_g$ and be aware of the change in the global rewards immediately if there are failures in the system.
However, the uncertainty in other agents' plans has also been shown to degrade the overall performance when using the global objective function to optimize local actions~\citep{wolpert2013probability}. To overcome this issue, we propose to let each agent improve its policy iteratively while assuming others keep their policies fixed.

More precisely, given a set of all possible action sequences of all agents $\Xc=(\Xc_n,\Xc_{-n})$, A-MCTS will periodically compute a ``best response" set of joint action sequences that maximize the joint utility for all participants $x^{BR}:= \{x^{BR}_n, x^{BR}_{-n}\}$ (Line 5). Each agent will then assume other agents coordinately determine their policies following such ``best response" $x^{BR}_{-n}$ and uses such information to compute the utility for its action sequence selection while growing the MCTS tree (Line 8). In general, the cardinality of $\mathcal{X}_n$ can be very large and it grows exponentially. To reduce the computation and communication requirements, we consider only those dynamically updated subsets $\mathcal{\hat{X}}_n \subseteq \mathcal{X}_n$ of the most promising action sequences. The set $\mathcal{\hat{X}}_n$ is chosen as the best rollouts of $M$ fixed nodes in the search tree $\mathcal{T}_n$ with the highest discounted empirical average (Line 3). We then define the following problem:

\begin{problem}\label{prob:2}(\textit{Best joint policy for multi-agent planning)}
\begin{equation}\label{eq:objectivefunctionRM}
     \underset{(x_1, x_2, \dots, x_N)}{\text{maximize}}\ U_g(x_1, x_2, ...,x_N)\
\end{equation}
 \begin{align}
        \text{Subject to:} \quad  &  x_n \in \hat{\Xc}_n,  \quad \forall n \in \mathcal{N}
	\end{align}
\end{problem}

The objective is to find an action profile $(x_1, x_2, \dots, x_{N})$ that maximizes the global utility $U_g(\cdot)$.
Such an optimization problem cannot be solved efficiently. Indeed it is NP-hard to maximize a submodular function~\citep{rezazadeh2023distributed}. Seeking a Nash equilibrium (NE) (where each agent policy is the best response to the others) that achieves a good efficiency compared to the optimal solution is more accessible~\citep{qu2019distributed}. A greedy algorithm is usually employed to find an approximation solution~\citep{nemhauser1978analysis}. However, we will show later with simulations that greedy solutions can be substantially suboptimal even in scenarios with few agents. In the following section, we provide a distributed regret-based solution to Problem~\ref{prob:2} that quickly and efficiently computes an NE joint policy for multi-agent systems, regardless of their complexity.

\subsection{Regret Matching For Cooperative Coordination}

In this paper, we consider the distributed solution of the optimization problem~\ref{prob:2} where each agent decides its path based on local information and limited communication from its peers. We aim to design a decision-making method that is capable of operating and adapting with occasional communication or less, where every agent acts solely based on its local observation and does not need to constantly communicate every decision with the others. This is to guarantee that the algorithm can effectively handle the agent attrition situation described in Section~\ref{sec:attrition_settings}. The main difficulty here is how to ensure the independent decisions of the agents lead to jointly optimal decisions for the group. To address this challenge, we formulate the problem of finding for each agent an action sequence that collectively maximizes the joint utility as a multi-agent cooperative game. We then propose a distributed mechanism, where every agent independently simulates a multi-player cooperative game based on the local information available to itself and solves the game by self-play. For this purpose, a game theory learning algorithm based on the Regret Matching technique~\citep{hart2000simple} is employed to approximate the Nash equilibrium of the game.

Let $\hat{\Xc}=(\hat{\Xc_n},\hat{\Xc}_{-n})$ denote the joint set of action sequences that are shared between all agents, and $x_{nm}$ denote the action sequence $m$ of agent $n$. In our approach, periodically, every agent independently constructs a matrix game in which the set of players contains all the active agents and the set of actions is $\hat{\Xc}$. At this stage, each agent applies the Regret Matching (RM) procedure as proposed in~\citep{hart2000simple} to its estimated matrix game to compute a best response joint decision. The pseudo-code of our RM game is shown in Algorithm~\ref{alg:RM}.

To further improve the performance of RM in cooperative settings, we let the agents use the global utility to calculate the regrets instead of the local utility. At each iteration $t$, an action $x_{nm} \in \Xc$ is sampled for each agent based on a probability distribution. Let $p$ denote this probability distribution where $p(x_{nm})$ is the probability for $x_{nm}$ and $\sum_{j=1}^{M} p(x_{nj}) = 1, \forall n \in \mathcal{N}$. With $x^{(t)} := \{x^{(t)}_n, x^{(t)}_{-n}\}$, we denote the sampled set at iteration $t$, where $x^{(t)}_n$ is the sample action for agent $n$ and $x^{(t)}_{-n}$ is the sampled actions for all agents except agent $n$. We then define the regret of agent $n$ for not taking action $m$ at iteration $t$ as
$R^{(t)}_{nm} = U_g(x_{nm}, x^{(t)}_{-n}) -  U_g(x^{(t)}).$
Denote $\Rc$ as the cumulative regret matrix where an element $\Rc_{nm}$ is the regret for $x_{nm}$ and $\Rc^+_{nm} = \max \{\Rc_{nm}, 0\}$. Then, the probability distribution $p$ used at the next iteration will be updated as
\begin{equation}
    p(x_{nm}) =
     \begin{cases} 
      \frac{\Rc^+_{nm}}{\sum_{m=1}^{M} \Rc^+_{nm}} & \text{if $\sum_{m=1}^{M} \Rc^+_{nm} > 0$,}\\
      \frac{1}{M} & \text{otherwise.}
    \end{cases}   \label{eq:prob_reget}
\end{equation}

\begin{algorithm}[!t]
 \caption{ \footnotesize Regret Matching Coordination algorithm \normalsize } %
 \label{alg:RM}
 \begin{algorithmic}[1]
 \renewcommand{\algorithmicrequire}
{\textbf{Input:}}
 \REQUIRE Global objective function $U_g$, joint compressed action sequences set $\hat{\Xc}$\\
 \renewcommand{\algorithmicrequire}{\textbf{Output:}}
 \REQUIRE Best response joint action sequences $x^{BR}$\\
 \STATE Every agent $n\in \Nc$ performs the following steps $2-9$ 
 \STATE \quad Initialize $\Rc$ to zeroes and $p$ to uniformly random
 \STATE \quad \textbf{for} $t=1,2, \dots$ \textbf{do}
 \STATE \quad \quad $x^{(t)} {\leftarrow} \textit{Sample} (\hat{\Xc}, p)$
 \STATE \quad \quad \textbf{for each} $x_{nm} \in \hat{\Xc}$ \textbf{do}
 \STATE \quad \quad \quad $\Rc_{nm}\ {\leftarrow} \Rc_{nm} + U_g(x_{nm}, x^{(t)}_{-n}) -  U_g(x^{(t)})$ 
 \STATE \quad \quad \quad Update $p(x_{nm})$ using Eq. (\ref{eq:prob_reget})

 \STATE \quad $x^{RM}_n {\leftarrow} \argmax_{x_{im} \in \hat{\Xc_i}} [p(x_{im})], \forall i \in \mathcal{N}$
 \STATE \quad $x^{RM}_{-n}\ {\leftarrow} \textit{Communicate and Update} (x^{RM}_n)$
 \STATE \textbf{return} $x^{BR} {\leftarrow} \argmax_{(x^{RM}_n,\ x^{RM}_{-n})} [U_g(x^{RM}_n,\ x^{RM}_{-n})]$
 \end{algorithmic}
\end{algorithm}

Denote the joint decision computed using RM by agent $n$, which is the set of action sequences, one per agent, that has the highest probability $p(x_{nm})$, as $x^{RM}_n$. Similarly, let $x^{RM}_{-n}$ be the computed set for all agents except agent $n$. Finally, these sets are exchanged between every agent, and the most payoff-dominant solution is chosen as the best response joint decision $x^{BR}=(x^{BR}_n,x^{BR}_{-n})$.

\subsection{Analysis and Discussion}
It has been shown in~\citep{berg2017exclusion} that there exists no polynomial time algorithm to compute a pure NE in multiplayer nonzero-sum stochastic games. Hence, we employ an approximate method of finding the NE by proposing a decentralized Nash selection method based on Regret Matching for making choices in a multiplayer matrix game formulated at each decision-making state. Regret Matching is a regret-based algorithm for learning strategies in games, and is often used to compute correlated equilibria in multi-player repeated games with imperfect information. Although the regret matching technique has been widely used for non-cooperative games, its application in cooperative games, such as the problem studied in our paper with a submodular utility function, has only been recently explored~\citep{nguyen2023social}. In this work, by leveraging the submodularity property of the joint objective function, we employ Regret Matching as a self-play technique to independently learn a Nash-based strategy for each player. We theoretically prove a stronger result of convergence using Regret Matching to an approximate pure-strategy Nash solution (see Definition~\ref{def:PSNE}), rather than the commonly-used correlated equilibrium, in games where players' utility functions are submodular.

\begin{definition}[Pure-Strategy Nash Equilibrium]\label{def:PSNE}
A pure-strategy Nash equilibrium (PSNE) is a joint action profile $x^*=(x^*_n,x^*_{-n}) \in \hat{\Xc}$ if for all $n\in \Nc$ and all $x_n \in \hat{\Xc}_n$ such that: 
$U_n(x^*_n,x^*_{-n}) \geq U_n(x_n,x^*_{-n})\ .$ 
\end{definition}

\begin{theorem}
\label{theorem:PSNE_convergence}
The best response joint decision $x^{BR}$ computed using RM, under the assumption of submodular utility functions, is a PSNE solution of the matrix-game representation generated by the set of best feasible paths $\hat{\Xc_n}\subseteq \Xc_n$ chosen by every agent at each decision point\footnote{For proof and analysis, please see Appendix B in the Supplementary Material }.
\end{theorem}

\section{Experimental Evaluation}

\begin{figure*}[!ht]
	\begin{center}
         \includegraphics[width=\linewidth]{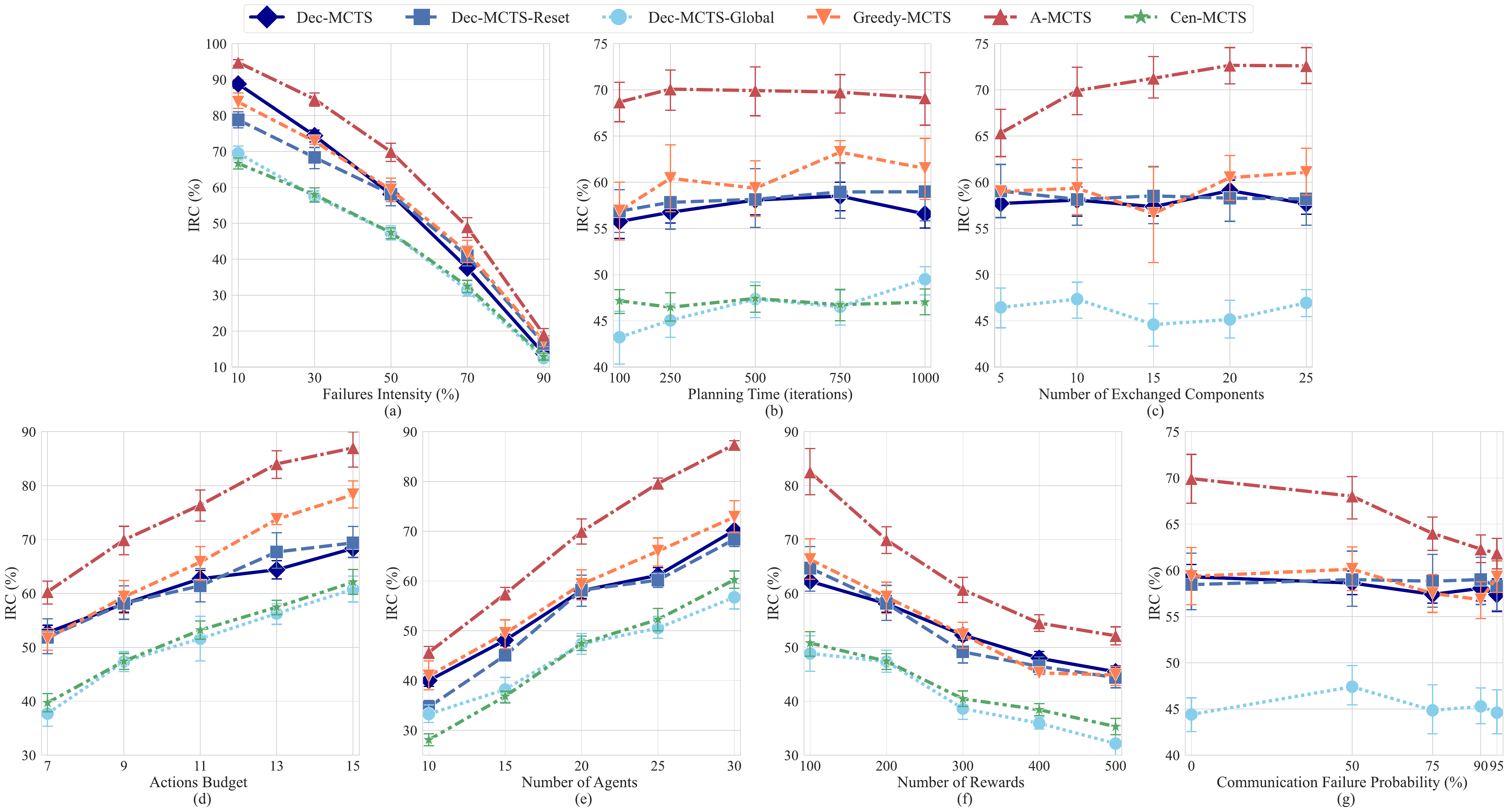}
		\caption{\footnotesize Impact of different parameters on the algorithms' performance at the mission end. Failures intensity (the fraction of agents that fail) (a); Planning time (b); Number of exchanged components (c); Actions budget (d), Number of agents (e); Number of rewards (f), and Communication failure probability (g). Results are with $95\%$ confidence interval.
        \label{fig:2}\normalsize 
        }
	\end{center}
\end{figure*}

To assess our A-MCTS algorithm, we consider the task of data collection from underwater wireless sensor networks (UWSN). Such tasks usually call for a collaboration of multiple autonomous underwater vehicles (AUVs) to traverse the environment and gather data from sensors.
The scenario consists of $200$ randomly distributed sensors in a $4000$ m $\times$ $4000$ m plane, with a transmission radius of $50$ m (typical for UWSN, e.g. \citep{cai2019data}). 
The graph $G$ of feasible paths is constructed using a probabilistic roadmap with a Dubins path model~\citep{kavraki1996probabilistic}. This model employs curves to refine the straight-line segments connecting waypoints and is extensively utilized for representing motion constraints pertinent to vehicle-like nonholonomic robots such as AUVs~\citep{41470}. The graph $G$ consists of $400$ vertices and an average of $19000$ edges.

We benchmark A-MCTS against the following baselines:
\begin{itemize}[leftmargin=*]
    \item \emph{Centralized MCTS (Central-MCTS)}: A single search tree is built for all of the $M$ harvesting agents with the actions of agent $m$ are at tree depth $(m, m + M, m + 2M,...)$.
    \item \emph{Dec-MCTS}~\citep{best2019dec}: It is the state-of-the-art decentralized multi-agent planning. In it, agents build their search tree with a marginal contribution utility function and adapt the same tree after churns occur.
    \item \emph{Dec-MCTS with reset (Dec-MCTS-Reset)}: Like Dec-MCTS, each agent builds its search tree with a marginal contribution utility function. Whenever churns occur, the tree of each agent is reset. This variant is considered to show that resetting the trees frequently could hamper the overall performance of the algorithm.
    \item \emph{Dec-MCTS with global utility (Dec-MCTS-Global}: Agents build their search tree with a global utility function and adapt the same tree after churns occur. This variant is considered to show that altering the utility function alone would not enhance performances against churns.
    \item \emph{A-MCTS with greedy optimization (Greedy-MCTS)}: In this scheme, we replace the RM Coordination (Algorithm 2) in our A-MCTS with a greedy algorithm, in which every agent sequentially picks the actions that deliver the highest immediate rewards for collaboration. This variant is considered to show that greedy solutions can be substantially suboptimal in MAS coordination.
\end{itemize}

For all algorithms, each planning phase consists of $500$ iterations, the discounting factor is set to $0.9$, and the exploration parameter is set to $0.4$ (i.e., within the ranges recommended in~\citep{best2019dec} to ensure the balance between exploration and exploitation). Unless otherwise stated, each agent compresses its tree into a set of $10$ possible paths and exchanges it with its teammates every $50$ planning iterations. Unless otherwise stated, we assume $20$ agents move in the graph, with a budget $B$ of $9$ actions. These values are chosen as they have proven to enable a high total utility score in the vast majority of scenarios considered in our experiments.

To model attrition in the population of agents, we assume that every agent has the same probability of failing during the mission and that the time at which each failure takes place is distributed uniformly at random throughout the mission duration. The key metric we use to evaluate the performance of the considered algorithms is the \emph{Instantaneous reward coverage (IRC)}, i.e. the fraction of available rewards covered (i.e. collected) at a given time.

\subsection{Performance Benchmarking}
In the first evaluation of our algorithm's adaptability to failures, we examine the impact of the failure intensity (i.e. fractions of agents that fail during the mission) on the IRC at the mission end, illustrated in Figure~\ref{fig:2}a. As expected, all algorithms experience performance declines with increasing intensity, reflecting reduced reward coverage due to fewer remaining agents in the system. Notably, with over $50\%$ of agents failing, Dec-MCTS-Reset surpasses the non-reset version due to the smaller system size which requires fewer iterations for exploration. Conversely, larger systems necessitate more time for agents to learn about the environment, hence frequent resets hamper the algorithm's performance.

To further elaborate on this matter, we assessed the impact of planning time on the IRC at the mission end. As Figure~\ref{fig:2}b shows, other baseline algorithms improved as planning time increased, with Dec-MCTS-Reset starting to outperform the non-reset with planning time larger than 750 iterations. A-MCTS, on the other hand, requires significantly less computational time yet still achieves the highest rewards regardless of the failure intensity and rates, thus proving itself as a good solution for online replanning.

The number of paths exchanged between agents significantly influences system complexity. Increased information exchange potentially leads to better algorithm performance, albeit at the expense of greater computational resources and time. To examine this trade-off, in Fig.~\ref{fig:2}c we evaluated the impact of different numbers of exchanged components on the IRC at the mission end. As expected, our proposed algorithm's performance improved with more exchanged components, while discounted algorithms showed no benefit. Indeed, with more exchanged components the utility of the joint policy found by regret matching improves too. However, given the finite number of optimal policies in a multi-agent game, escalating the number of components eventually yields diminishing returns.

As the above results show, resetting the tree would not consistently lead to improvement because the planning process involves initial exploration in which agents take random actions to learn reward distribution. Resetting without sufficient planning time results in suboptimal joint policies. Additionally, the use of the \emph{marginal contribution} utility combined with a \emph{submodular} reward function hampers agents' ability to recognize global reward reduction and hence adapt to failures efficiently. Moreover, sampling other agents' action sequences increases variance in global utility estimation and degrades the coordination quality. By assuming that the policies of other agents are fixed, A-MCTS can overcome this instability issue and adapt to agent failures better. Indeed, with regret matching aids in discovering better joint policies and thus provides better guidance for the exploration-exploitation of the search tree, our method exhibits superior performances in all cases.

In another set of experiments, we evaluated the impact of action budget $B$, the number of agents $N$, the number of rewards, and the communication failure probability, for a default failure intensity of $50\%$. As Fig. \ref{fig:2}d shows, A-MCTS managed to outperform Dec-MCTS substantially despite the difficulty of decentralized planning with a growing actions budget. Furthermore, as we doubled the budget of the action, the superiority of A-MCTS over the other techniques tripled. A similar behavior is exhibited by the system when we vary the number of agents. As shown in Fig.~\ref{fig:2}e, with a small number of agents, the differences between our algorithm and the discounted methods grow to $20\%$ with an increasing number of agents. In the considered settings, we also assess the impact of the density of sensors on the algorithms' performance by varying the number of sensors within the same area. As shown in Fig.~\ref{fig:2}f, the IRC declines as more sensors are introduced in the system. Naturally, with an increasing number of sensors the area that must be covered by agents expands as well. Nevertheless, A-MCTS shows better scalability as it consistently outperforms other methods.

The effectiveness of cooperative MAS is notably influenced by the quality of inter-agent communication. To understand better the impact of such limitations, we evaluated the algorithms' performances under different communication failure probabilities between each agent pair during a mission. As shown in Fig.~\ref{fig:2}g, there is no notable decline in A-MCTS performance even when half of the communication is disrupted, and it continues to outperform baseline methods with severely unstable communication. This highlights A-MCTS's ability to enable efficient cooperation among agents in hostile environments with restricted communication.

\subsection{Trade Off Between Communication Loss and Attrition for A-MCTS Analysis}

\begin{figure}[!ht]
	\begin{center}
        \includegraphics[width=\linewidth]{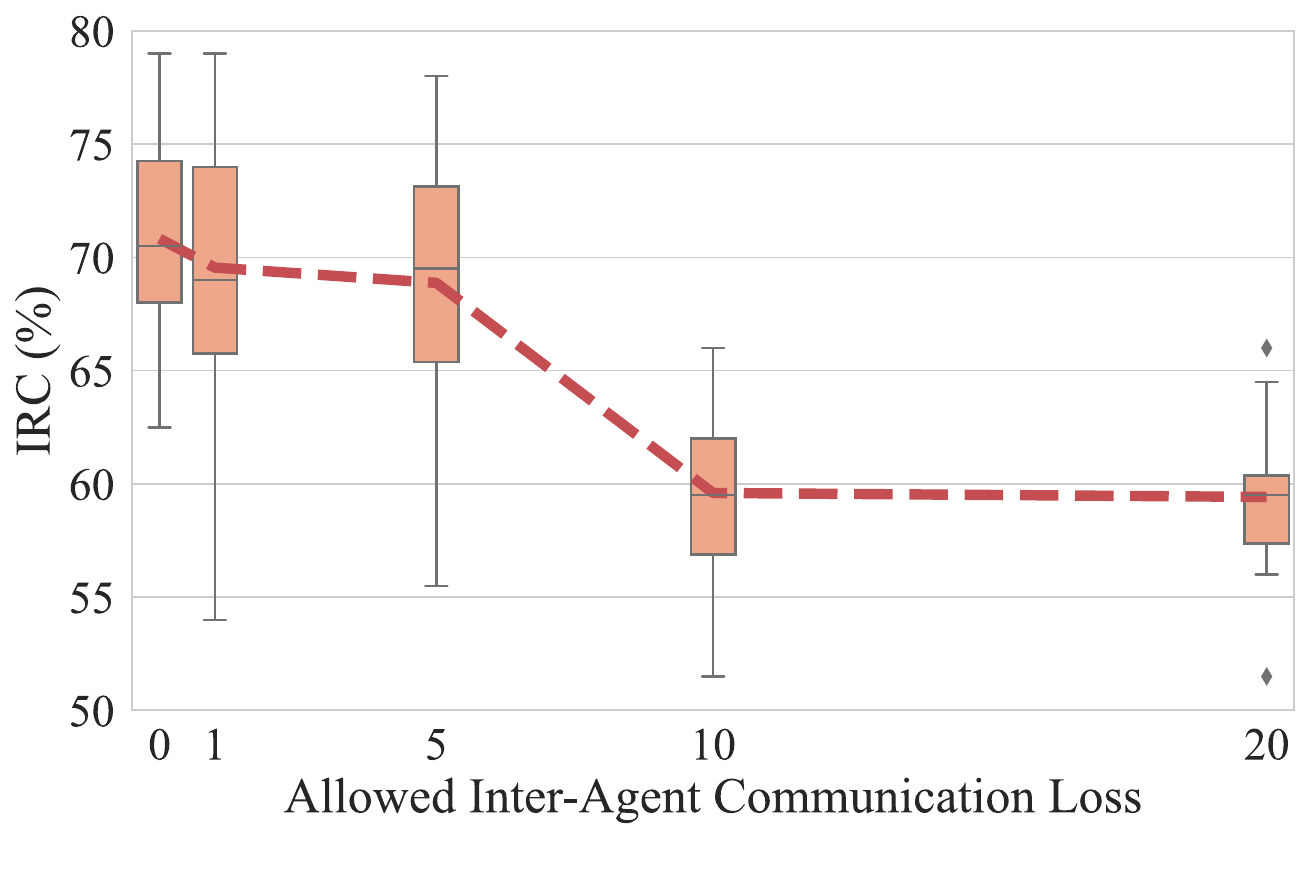}
		\caption{Impact of allowed inter-agent communication loss on the performance of A-MCTS at the end of the mission.
        \label{fig:a2}\normalsize}
	\end{center}
\end{figure}

In our approach, repeated communication loss is used as an indication of attrition. However, in practical settings, inter-agent communication can be unreliable and intermittent. If the algorithm is more sensitive to communication loss, it can mistakenly treat delayed messages as agent failures.

To better understand such impact, in this section, we study how tolerance to communication loss affects the performance of A-MCTS. Specifically, we parameterize this tolerance level by the number of times an agent must experience communication loss with another agent before treating it as attrition. Fig.~\ref{fig:a2} shows the IRC at the end of the mission against different numbers of allowed inter-agent communication loss. With up to 5 allowed messages loss, A-MCTS still shows no noticeable degradation. However, as the algorithm is more communication loss tolerant, the performance declines. This is expected because the remaining agents can not recognize churns fast enough and adapt efficiently.

\section{Conclusions}
Achieving efficient coordination in multi-agent planning for information gathering is a critical challenge in practical settings with high attrition rates. In this work, we proposed a new approach to tackle this issue. Our proposed algorithm, Attritable MCTS (A-MCTS), effectively coordinates actions among agents while adapting to agent failures by allowing all agents to jointly optimize the global utility directly with a new coordination technique based on regret matching. Our empirical evaluation demonstrates that A-MCTS improves substantially over the best existing approaches in terms of global utility and scalability in scenarios with frequent agent failures. As a follow-up, we intend to extend A-MCTS to more dynamic systems where new agents can join and the communication is probabilistic. Another line of inquiry is to study the performance of our algorithm in problems with inter-agent dependency, where the actions of an agent can only be enabled by the actions of another.


\clearpage
\bibliography{mybibfile}
\newpage
\appendix

\section{Technical Results}
\label{app:A}
\subsection{Details of Discounted Upper Confidence Bound on Tree (D-UCT)}

Consider an arbitrary node $s$ with a set of child nodes $\Ac_n(s)$. Whenever $s$ is visited, the child node $s^{\prime}\in \Ac_n(s)$ with the largest D-UCB is chosen as 
\begin{equation}
    a_n^{(t)}(s) =\argmax_{s^{\prime} \in \Ac_n(s)} X_n^{(t)}(s^{\prime})\ ,
\end{equation}
where $X_n^{(t)}(s^{\prime})$ is the D-UCB score for node $s^{\prime}$ at iteration $t$. $X_n^{(t)}(s^{\prime})$ is updated using the D-UCB algorithm~\citep{garivier2011upper} as follows. 

First, let $\gamma \in (1/2,1)$ be a discounting factor and $C_p>1\sqrt{8}$ an exploration constant, the upper bound confidence for child node $s^{\prime}$ is calculated as
\begin{equation}
     X_n^{(t)}(s^{\prime}, \gamma) \coloneqq \bar{F}_n^{(t)}(s^{\prime}, \gamma) + c_n^{(t)}(s^{\prime}, \gamma)\ ,
\end{equation}
where $\bar{F}_n^{(t)}(s^{\prime},\gamma)$ is the average empirical reward for choosing $s^\prime$, and $c_n^{(t)}(s^\prime, \gamma)$ is the exploration bonus. 

Denote the discounted number of times the child node $s^{\prime}$ of the parent node $s$ has been visited as
\begin{equation}
\label{eq:Nt}
   N_n^{(t)}(s^{\prime},\gamma) \coloneqq \sum\nolimits_{\tau=1}^t \gamma^{t-\tau} \textbf{1}_{ \left\{ a_n^{(\tau)}(s)=s^{\prime} \right\} }\ , 
\end{equation}
where $\textbf{1}_{ \left\{ a_n^{(\tau)}(s)=s^{\prime} \right\} }$ is the indicator function that returns $1$ if node $s^{\prime}$ was selected at round $\tau$ and $0$ otherwise.

Let $F_n^{(\tau)}$ be the rollout score at iteration $\tau \leq t$ and $N_n^{(t)}(s,\gamma)$ be the discounted number of times the parent node $s$ has been visited. Then at time $t$, the average reward for node $s^\prime$ is computed as
\begin{equation}
\label{eq:rollout_rewards}
\bar{F}_n^{(t)}(s^\prime,\gamma) = \frac{1}{N_n^{(t)}(s^{\prime},\gamma)} \sum\nolimits_{\tau=1}^{t} \gamma^{t-\tau}F_n^{(\tau)} \ \textbf{1}_{ \left\{ a_n^{(\tau)}(s)=s^{\prime} \right\} },
\end{equation}

and exploration bonus as
\begin{equation}
\label{eq:exploration}
c_n^{(t)}(s^\prime, \gamma) = 2C_p \sqrt{\frac{\log N_n^{(t)}(s,\gamma) }{N_n^{(t)}(s^{\prime},\gamma)}}.
\end{equation} 

\subsection{Analysis of Dec-MCTS Performance in Attrition Settings}

As shown in~\citep{best2019dec}, Dec-MCTS has vanishing regret and converges as $t \rightarrow \infty$. We prove here the behavior of Dec-MCTS after convergence. Recall that by convergence we mean each agent stays with the same action sequence (i.e., the UCB score for each action in such sequence is the highest at that corresponding decision node).

Assume that Dec-MCTS converges at iteration $\tau_0$ (finite) for all agents. At iteration $t > \tau_0$, let $x^*_n$ denote the converged action sequence of agent $n$, and $x^*_{-n}$ denotes the converged action sequences of every other agent except agent $n$. The rollout score received by agent $n$ for executing the action sequence $x^*_n$ given by the \emph{marginal utility }will then be a constant $L$:
\begin{equation}
\label{eq:converged_rollout}
F_n^{(t)}(x^*_n) = U_n(x^*_n,x^*_{-n}) = U_g(x^*_n,x^*_{-n}) - U_g(x^*_{-n}) = L \ .
\end{equation}

Assume that at the next iteration $t+1$, a subset of agents becomes unavailable due to failures. Let $x^\prime_{-n}$ be the combined sequences of actions taken by all agents except agent $i$ and the lost agents. That is
$$x^\prime_{-n} \subseteq x^*_{-n}.$$
and the rollout score agent $i$ receives for the same action sequence now is
\begin{equation} \label{eq:A}
    F_n^{(t+1)} (x^*_n) = U_n(x^*_n,x^\prime_{-n}) = U_g(x^*_n,x^\prime_{-n}) - U_g(x^\prime_{-n})\ .
\end{equation}

\begin{lemma}\label{le:suboptimality}
Assume that the Dec-MCTS algorithm has converged on all the agents at time step $\tau_0$ and that the global objective function $U_g$ is submodular. Then
$$X_n^{(t+1)}(x^*_n,\gamma) \geq X_n^{(t)}(x^*_n,\gamma),\ \forall t \geq \tau_0.$$ 
\end{lemma}

Lemma \ref{le:suboptimality} essentially states that once Dec-MCTS converges, the D-UCB score calculated by agent $n$ for its converged action sequence $x_n$ is non-decreasing even if it detects that some of the other agents have failed. Hence, it always picks and updates the same action sequence (i.e., the series of actions that has the highest D-UCB scores at each corresponding decision node).

Let $c, p \in x^*_n$ be two nodes in the converged action sequence of agent $i$, with $c$ being the child node of $p$. After the algorithm converges at iteration $\tau_0$, by definition, the nodes $c$ and $p$ are going to be selected repeatedly. Thus, at iteration $t$, the discounted number of times $c$ is visited can be written as
\begin{eqnarray} \label{N_n_c}
\begin{aligned}
    N^{(t)}_n(c,\gamma) &= \gamma^{t-\tau_0}\ N^{(\tau_0)}_c + \sum_{\tau=0}^{t-\tau_0} \gamma^{\tau}  \\ 
    &= \gamma^{t-\tau_0}\ N^{(\tau_0)}_c + \frac{1-\gamma^{t-\tau_0+1}}{1-\gamma}, 
\end{aligned}
\end{eqnarray}
with the constant $N^{(\tau_0)}_c$ is the discounted number of times node $c$ is chosen at $\tau_0$. In addition, the discounted number of times $p$ is visited at iteration $t$ is
\begin{eqnarray} \label{N_n_p}
\begin{aligned}
    N^{(t)}_n(p,\gamma) &= \gamma^{t-\tau_0}\ N^{(\tau_0)}_p + \sum_{\tau=0}^{t-\tau_0} \gamma^{\tau} \\ 
    &= \gamma^{t-\tau_0}\ N^{(\tau_0)}_p + \frac{1-\gamma^{t-\tau_0+1}}{1-\gamma}, 
\end{aligned}
\end{eqnarray}
with the constant $N^{(\tau_0)}_p$ is the discounted number of times node $p$ is chosen at $\tau_0$. Finally, the accumulated rollout score for $c$ at iteration $t$ is
\begin{eqnarray} \label{F_n_c}
\begin{aligned}
    \sum_{\tau=1}^t \gamma^{t-\tau} F^{(\tau)}_n \textbf{1}_{ \left\{ a_n^{(\tau)}(p)=c \right\}} &= \gamma^{t-\tau_0} F^{(\tau_0)}  + L\ \sum_{\tau=0}^{t-\tau_0} \gamma^{\tau} \\ 
    &= \gamma^{t-\tau_0} F^{(\tau_0)}  + L\ \frac{1-\gamma^{t-\tau_0+1}}{1-\gamma}\ ,
\end{aligned}
\end{eqnarray}

with the constant $F^{(\tau_0)}$ is the accumulated rollout score for $c$ at $\tau_0$, and $L$ is the rollout score for $c$ at every iterations up to $\tau_0$ as given in (\ref{eq:converged_rollout}). For brevity of notations, we denote the following values
\begin{eqnarray} \label{eq:NPCF}
    \begin{aligned}
        &N_c = \gamma^{t-\tau_0}\ N^{(\tau_0)}_c + \frac{1-\gamma^{t-\tau_0+1}}{1-\gamma}, \\
        &N_p = \gamma^{t-\tau_0}\ N^{(\tau_0)}_p + \frac{1-\gamma^{t-\tau_0+1}}{1-\gamma}, \\
        &F_c = \gamma^{t-\tau_0}\ F^{(\tau_0)}  + L\ \frac{1-\gamma^{t-\tau_0+1}}{1-\gamma}, \\
        & A  = F_n^{(t+1)} (x^*_n).
    \end{aligned}
\end{eqnarray}

At iteration $t+1$ when failures occur, the values for the discounted numbers of times node $c$ and $p$ are chosen, and the accumulated rollout score for $c$ can be updated as
\begin{eqnarray} \label{t+1}
\begin{aligned}
    &N^{(t+1)}_n(c,\gamma) = N_c + \gamma^t, \\
    &N^{(t+1)}_n(p,\gamma) = N_p + \gamma^t, \\ 
    &\sum_{\tau=1}^{t+1} \gamma^{t+1-\tau} F^{(\tau)}_n \textbf{1}_{ \left\{ a_n^{(\tau)}(p)=c \right\}} = F_c\ \gamma + A\ ,
\end{aligned}
\end{eqnarray}

\noindent with $A$ is the rollout score for $c$ at iteration $t+1$ as given above. The inequality of Lemma \ref{le:suboptimality} now can be written as
\begin{eqnarray} \label{f_t}
\begin{aligned}
 &\frac{F_c\gamma + A}{N_c + \gamma^t} + c_p \sqrt{\frac{2\log(N_p + \gamma^t)}{N_c + \gamma^t}} \geq \frac{F_c}{N_c} + c_p \sqrt{\frac{2\log(N_p)}{N_c}} \\
 \Leftrightarrow &\frac{F_c\gamma + A}{N_c + \gamma^t} - \frac{F_c}{N_c}&  \\
 &+ c_p\left(\sqrt{\frac{2\log(N_p + \gamma^t)}{N_c + \gamma^t}} - \sqrt{\frac{2\log(N_p)}{N_c}}\right) \geq 0.
\end{aligned}
\end{eqnarray}

Let 
\begin{equation}
f(t) = \frac{F_c\gamma + A}{N_c + \gamma^t} - \frac{F_c}{N_c} + c_p\left(\sqrt{\frac{2\log(N_p + \gamma^t)}{N_c + \gamma^t}} - \sqrt{\frac{2\log(N_p)}{N_c}}\right)    
\end{equation}
be a funtion of time $t$ over the set of fixed paramters $\{\gamma, c_p, F_c, N_c, N_p\}$.

It can be verified that $f(t)$ is an increasing function as the derivative of $f(t)$ is positive for $t \gg \tau_0$. In addition, as $t \gg \tau_0$, the inequality of (\ref{f_t}) becomes:
\begin{eqnarray}
\begin{aligned}
 \frac{F_c\gamma + A}{N_c} &\geq \frac{F_c}{N_c} \\ \nonumber
 \Leftrightarrow \quad \quad \quad A &\geq F_c(1-\gamma) \\ \
 \Leftrightarrow \quad \quad \quad A &\geq \left(\gamma^{t-\tau_0} F^{(\tau_0)}  + L\frac{1-\gamma^{t-\tau_0+1}}{1-\gamma}\right)(1-\gamma) \\
 \Leftrightarrow \quad \quad \quad  A &\geq L.
\end{aligned}
\end{eqnarray}

The last inequality follows from the assumption that the global utility function $U_g$ is submodular. That is, having failures as time $t+1$ implies there are fewer agents collecting rewards, hence the local utility for agent $i$ increases (or remains the same). Thus $A \geq L$.

Since Proposition 1 gives that $ F_n^{\tau+1}(x^*_n) \ge F_n^{\tau}(x^*_n)$, there exists a $\tau_0$ for which $f(t)$ is non-negative for some $t \gg \tau_0$. This implies the UCB scores for each node in the actions sequence $x^*_n$ will remain the highest. Thus, agent $i$ will continue to select $x^*_n$  after failures. This concludes the proof.

\begin{figure}[!ht]
	\begin{center}
		\includegraphics[width=\linewidth]{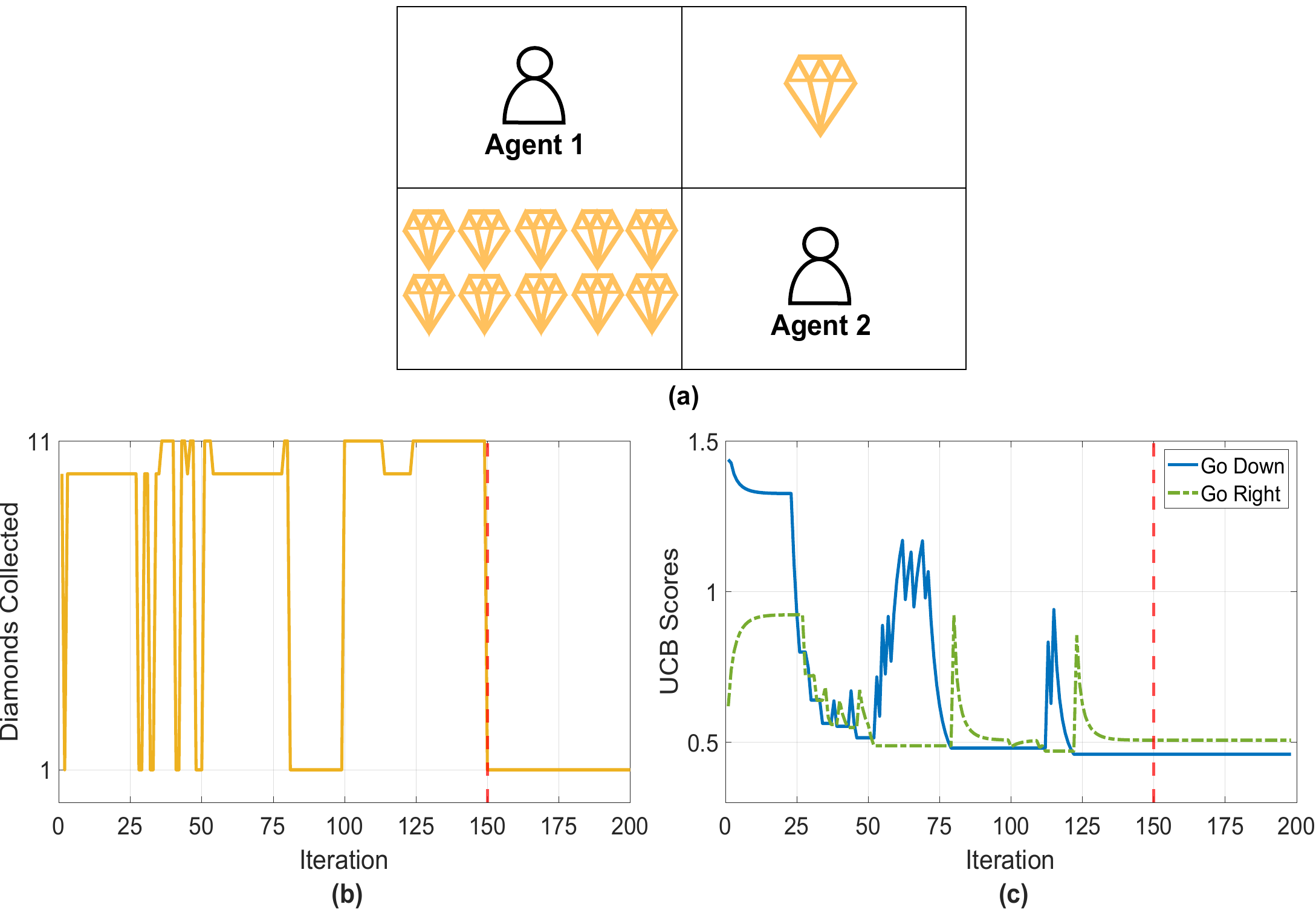}
		\caption{\small  Diamonds collection game (a), number of diamonds collected (b), and D-UCB score for each action of Agent 1 (c).
            \label{fig:intro}\normalsize}
	\end{center}
\end{figure}

Let's illustrate the significant implication of the lemma through an example. Consider a grid-world diamonds collection game~\citep{sewak2019deep}, where two agents play in a team using Dec-MCTS. An exploration factor $C_p$ and a discounting factor $\gamma$ for Dec-MCTS are chosen as $0.5$ and $0.75$ respectively. As shown in Figure~\ref{fig:intro}, when simulating the game we see that the D-UCB scores for Agent 1 fluctuate until it converges to $\{0.4607, 0.5072\}$ (after $135$ iterations in our example). The empirical average reward of \emph{Go Right} is estimated by Agent 1 by dividing its contribution (one diamond) to the global utility ($11$ diamonds), while the discounted number of visits is approximately $4$, yielding the asymptotic score of $0.5072$. The D-UCB score for \emph{Go Down} (the sub-optimal action) is non-deterministic depending on the random choices made by the two agents during the initial transient. This value is not updated in convergence as that branch of the search tree is not sampled, due to the MCTS selection policy.

Using the marginal contribution as the utility improves stability and convergence speed, but it causes issues when agents fail, as shown. At iteration $150$, Agent 2 fails (or leaves the game). In this case, even if the optimal choice for Agent 1 would be \emph{Go Down} (due to the higher amount of diamonds), it sticks with \emph{Go Right}. This happens because both the exploration bonus and the local contribution to the overall \emph{hypothetical global reward} remain the same, despite the real global reward has been reduced.

\section{Proof of Theorem 1}
\label{app:B}
Before going through the proof of Theorem 1, we first prove the existence of at least one PSNE in the formulated game.

\begin{lemma}\label{Existence_PSNE}
A finite coordination game will always have at least one PSNE, if maximizing players' local utilities corresponds to maximizing the global objective, i.e., the players' local utility functions satisfy, $\forall x_n, x^{\prime}_n\in\hat{\Xc}_n,\ \forall x_{-n}\in\hat{\Xc}_{-n},\ \forall n\in \Nc\,$
\begin{eqnarray}
\begin{aligned} \label{eq:potential_function}
    U_n(x_n,x_{-n})-U_n(x^{\prime}_n,x_{-n}) &> 0 \\
    \Rightarrow\ \Phi(x_n,x_{-n})-\Phi(x^{\prime}_n,x_{-n}) &> 0
\end{aligned}
\end{eqnarray}
\noindent where $\Phi(\cdot)$ is a function that represents the global objective.
\end{lemma}

\begin{proof}
Every finite coordination game in which the global objective function is aligned with the local utility functions of the players, that is, satisfies the property as in~\eqref{eq:potential_function}, is a generalized ordinal potential game~\citep{marden2009cooperative}.
Let $\phi$ be a potential function of a coordination game $\Gc$. Then the equilibrium set of $\Gc$ corresponds to the set of local maxima of $\phi$. That is, an action profile $x =(x_n,x_{-n})$ is a NE point for $\Gc$ if an only if for every $n\in\Nc$,
$$\phi(x) \geq \phi(x^{\prime}_n, x_{-n}),\ \forall x^{\prime}_n\in \hat{\Xc}_n \ .$$

Consider $x^*=(x_n^*,x_{-n}^*)\in\hat{\Xc}$ for which $\phi(x^*)$ is maximal (which is true by definition for a finite set $\hat{\Xc}$), then for any $x^{\prime}=(x_n^{\prime},x_{-n})$:
\begin{equation*}
\phi(x_n^*,x_{-n}^*) > \phi(x^{\prime}_n,x_{-n}) 
\Leftrightarrow U_n(x_n^*,x_{-n}^*) > U(x^{\prime}_n,x_{-n})\ .
\end{equation*}
Hence, the game possesses a pure strategy NE.
\end{proof}

We now proceed with the main proof. It is well known that, for any finite matrix game, if all players apply the same Regret Matching policy the empirical distribution of all players' joint action converges to the set of {\it Coarse Correlated Equilibria (CCE)}~\citep{hart2000simple}. We prove a stronger result of convergence to a PSNE under the assumption of submodular utility functions. 

\begin{figure*}[!ht]
	\begin{center}
        \includegraphics[width=\linewidth]{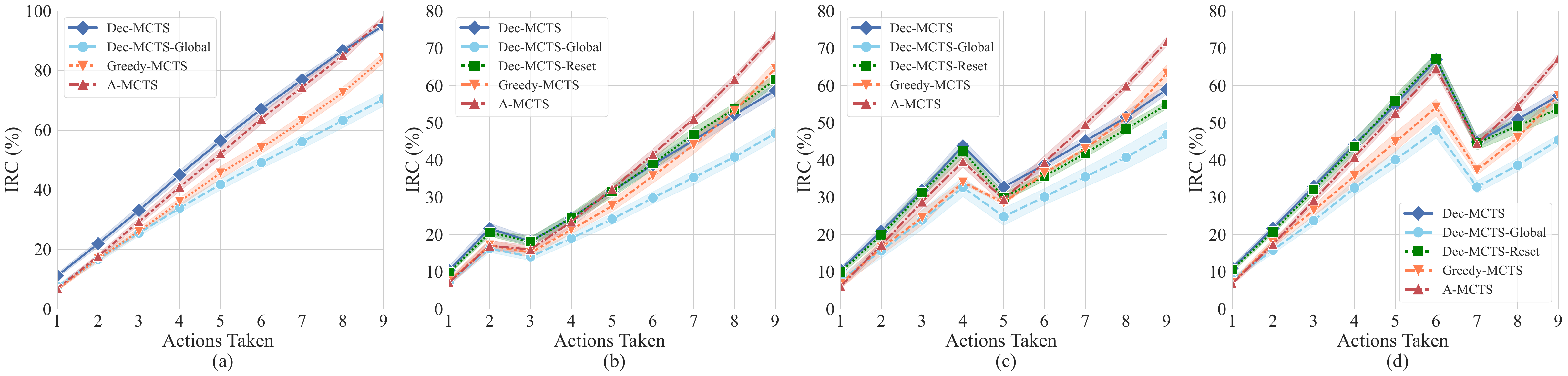}        
		\caption{Evolution over the mission of the Instantaneous Reward Coverage (IRC) in the \emph{Forced Failure} setting for different times of attrition: no attrition (a), attrition after 2 actions (b), attrition after 4 actions (c), and attrition after 6 actions (d). Results are with $95\%$ confidence.\label{fig:a1}}
	\end{center}
\end{figure*}

As we formulate the problem of multi-agent information gathering as maximization of a submodular function, the considered matrix game generated by the active agents and their corresponding sets of best feasible paths at each decision point satisfies the following two properties: 
\begin{itemize}
    \item \emph{Property 1}: $\displaystyle\sum_{n\in\Nc} \lambda_n\ U_n(x)$ is concave in $x$,
    \item \emph{Property 2}: $U_n(x_n,x_{-n})$ is convex in $x_{-n}$,
\end{itemize}
where $x:=(x_n,x_{-n})$ denotes a pure joint action in which agent $n$ chooses path $x_n$ and the other agents select $x_{-n}$.
The combination of the two properties implies that player $n$'s local utility function $U_n(\cdot)$ is concave in $x_n$ given $x_{-n}$ is fixed. 

Let $x$ be a CCE of the considered game, and let $\bar{x} = \Er_{\pi} \lsb x\rsb$, we then prove that $\bar{x}$ is a pure strategy NE of the game. Without loss of generality, assume that $\lambda_n = 1,\ \forall n \in N$. As $x$ is a CCE point, it satisfies
\begin{equation} \label{eq:CCE_property}
\Er \lsb U_n(x) \rsb \geq \Er \lsb U_n(x_n^{\prime},x_{-n}) \rsb \ ,
\end{equation}
for every $n \in N$ and every action $x_n^{\prime}\in\hat{\Xc}_n$. Also, since $\bar{x}\in \hat{\Xc}$, using Property 2 we have
\begin{equation} \label{eq:CCE_property2}
\Er \lsb U_n(x_n^{\prime}, x_{-n} \rsb \geq U_n \big( x_n^{\prime}, \Er \lsb x_{-n} \rsb \big) = U_n(x_n^{\prime}, \bar{x}_{-n})\ .
\end{equation}

Combining~\eqref{eq:CCE_property} and~\eqref{eq:CCE_property2} yields
\begin{equation} \label{eq:CCE_property3}
\Er \lsb U_n(x) \rsb \geq U_n(x_n^{\prime}, \bar{x}_{-n}) \ .
\end{equation}

Replacing $x_n^{\prime} = \bar{x}_n$ and then summing over all $n\in N$
\begin{equation}
\sum_{n\in N} \Er \lsb U_n(x) \rsb \geq \sum_{n \in N} U_n(\bar{x}_n, \bar{x}_{-n}) = \sum_{n\in N} U_n(\bar{x}) \nonumber.
\end{equation}

Using Property 1 implies 
\begin{equation}
\sum_{n \in N} \Er \Big[ U_n (x) \Big] = \Er \lsb \sum_{n \in N} U_n(x) \rsb \leq \sum_{n \in N} U_n \Big( \Er \lsb x \rsb \Big)\ \nonumber.
\end{equation}

Therefore
$$\sum_{n\in N} \Er \Big[ U_n (x) \Big] = \sum_{n \in N}U_n(\bar{x}) \ .$$

Thus, $U_n(\bar{x})=\Er\lsb U_n(x)\rsb$ for every $n$, and~\eqref{eq:CCE_property3} becomes
$$U_n(\bar{x})\geq U_n(x_n^{\prime}, \bar{x}_{-n})\ .$$
for every $x^{\prime}_n\in\hat{\Xc}_n$. Therefore, $\bar{x}$ is a pure Nash equilibrium. This implies that the time average of the joint action of all players converges to a PSNE solution.

\section{Additional Experimental Resuls}
\label{app:D}
\subsection{Time of Attrition Analysis}

To perform a baseline evaluation of our algorithm, we consider a setting with no attrition and measure the task performance in terms of instantaneous reward coverage throughout the mission. As shown in Fig.~\ref{fig:a1}a, under this static environment, although the IRC of A-MCTS appeared to be less than Dec-MCTS initially, it ended up comparable and even slightly outperformed the state-of-the-art at the end of the mission. This shows that our proposed algorithm can discover paths that guarantee more long-term rewards and thus is also a good fit for multi-agent coordinated information gathering in general settings. 

To evaluate our algorithm's adaptability to failures, we considered the \emph{Forced Failure} setting, in which after a specific number of actions have been taken, half of the agents (chosen at random) become unavailable. Specifically, Fig.~\ref{fig:a1}b, c, and d shows the instantaneous reward coverage with attrition at the early stage (e.g., after 2 actions), middle stage (e.g., after 4 actions), and later stage (e.g., after 6 actions) of the mission respectively. As these figures show, resetting the tree for replanning produced no significant benefits compared to those that adopted the same tree. This is because every MCTS process starts with the exploration phase where agents intentionally take random actions to learn the reward distribution. As such, resetting the tree without sufficient planning would cause the produced joint policy from this period to be sub-optimal. In addition, as Dec-MCTS uses the \emph{marginal contribution} as the utility function, it is unable to recognize the reduction of the global reward and hence is unable to adapt to failures efficiently. Indeed, the gap between it and Dec-MCTS-Global is halved compared to the case with no failure. However, using the global utility function alone is not enough, as sampling other agents' action sequences introduces a lot of variance in the estimation of the global utility. By assuming that the policies of other agents are fixed, both A-MCTS and Greedy-MCTS can overcome this instability issue and adapt to agent failures better, with A-MCTS performing the best in all cases as the regret matching method allowing the agents to discover better joint policies and thus provides better guidance for the exploration-exploitation of the search tree. It is also interesting to note that the superiority of A-MCTS compared to Dec-MCTS is slightly reduced (from $15\%$ to $10\%$) as attrition occurs later. This is expected as when some agents failed in the final stage of the mission, the remaining agents would not have enough action budget left to recover the lost rewards.

\end{document}